\documentclass[twoside, journal]{IEEEtran}

\usepackage{epsfig}
\usepackage{graphicx}
\usepackage{amsmath}
\usepackage{amssymb}
\usepackage{float}
\usepackage{url}
\newtheorem{lemma}{Lemma}

\newtheorem{theorem}{Theorem}

\newcommand{\naturals}{\ensuremath{\mathbb{N}}}
\newcommand{\reals}{\ensuremath{\mathbb{R}}}
\newcommand{\pr}{\ensuremath{\mathbb{P}}}
\newcommand{\expectation}{\ensuremath{\mathbb{E}}}

\begin{document}

\title{\huge{Moderate Deviations Analysis of Binary Hypothesis Testing}}

\author{\IEEEauthorblockN{Igal Sason\\
\hspace*{-0.4cm} \url{sason@ee.technion.ac.il}\\
Department of Electrical Engineering\\ Technion, Haifa 32000,
Israel}}

\maketitle\thispagestyle{empty}

\begin{abstract}
This work refers to moderate-deviations analysis of
binary hypothesis testing. It relies on a concentration
inequality for discrete-parameter martingales with bounded jumps,
which forms a refinement to the Azuma-Hoeffding
inequality. Relations of the analysis to the moderate deviations
principle for i.i.d. random variables and the relative entropy
are considered.
\end{abstract}


\section{Introduction}
\label{section: introduction} The moderate deviations analysis in
the context of source and channel coding has recently attracted
some interest among information theorists (see \cite{Abbe_thesis},
\cite{AltugW_ISIT2010}, \cite{He_IT09},
\cite{Polyanskiy_Poor_Verdu_IT2010}, \cite{Sason_submitted_paper}
and \cite{Tan_arxiv11}). The purpose of this paper is to consider
moderate deviations analysis for binary hypothesis testing.

In the following, related literature on moderate deviations
analysis in information-theoretic aspects is shortly reviewed.
Moderate deviations were analyzed in
\cite[Section~4.3]{Abbe_thesis} for a channel model that gets
noisier as the block length is increased. Due to the dependence of
the channel parameter in the block length, the usual notion of
capacity for these channels is zero. Hence, the issue of
increasing the block length for the considered type of degrading
channels was examined in \cite[Section~4.3]{Abbe_thesis} via
moderate deviations analysis when the number of codewords
increases sub-exponentially with the block length. In another
recent work \cite{AltugW_ISIT2010}, the moderate deviations
behavior of channel coding for discrete memoryless channels was
studied by Altug and Wagner with a derivation of direct and
converse results which explicitly characterize the rate function
of the moderate deviations principle (MDP). In
\cite{AltugW_ISIT2010}, the authors studied the interplay between
the probability of error, code rate and block length when the
communication takes place over discrete memoryless channels,
having the interest to figure out how the decoding error
probability of the best code scales when simultaneously the block
length tends to infinity and the code rate approaches the channel
capacity. The novelty in the setup of their analysis was the
consideration of the scenario mentioned above, in contrast to the
case where the rate is kept fixed below capacity, and the study is
reduced to a characterization of the dependence between the two
remaining parameters (i.e., the block length $n$ and the average/
maximal error probability of the best code). As opposed to the
latter case when the code rate is kept fixed, which then
corresponds to large deviations analysis and characterizes the
error exponents as a function of the rate, the analysis in
\cite{AltugW_ISIT2010} (via the introduction of direct and
converse theorems) demonstrated a sub-exponential scaling of the
maximal error probability in the considered moderate deviations
regime. This work was followed by a work by Polynaskiy and
Verd\'{u} where they show that a DMC satisfies the MDP if and only
if its channel dispersion is non-zero, and also that the AWGN
channel satisfies the MDP with a constant that is equal to the
channel dispersion. The approach used in \cite{AltugW_ISIT2010}
was based on the method of types, whereas the approach used in
\cite{Polyanskiy_Verdu_Allerton2010} borrowed some tools from a
recent work by the same authors in
\cite{Polyanskiy_Poor_Verdu_IT2010}.

In \cite{He_IT09}, the moderate deviations analysis of the
Slepian-Wolf problem for lossless source coding was studied. More
recently, moderate deviations analysis for lossy source coding of
stationary memoryless sources was studied in \cite{Tan_arxiv11}.

These works, including this paper, indicate a recent interest in
moderate deviations analysis in the context of
information-theoretic problems. In the literature on probability
theory, the moderate deviations analysis was extensively studied
(see, e.g., \cite[Section~3.7]{Dembo_Zeitouni}), and in particular
the MDP was studied in \cite{Dembo_paper96} for continuous-time
martingales with bounded jumps.

This paper has the following structure: Section~\ref{section:
Concentration Inequalities via Martingales} introduces briefly
some preliminary material related to martingales and Azuma's
inequality. It then follows by introducing a refined version of Azuma's
inequality, and a study of its relation to the moderate
deviations principle for i.i.d. random variables.
Section~\ref{section: binary hypothesis testing} considers the
relation of Azuma's inequality and the refined version of this
inequality (from Section~\ref{section:
Concentration Inequalities via Martingales}) to moderate
deviations analysis of binary hypothesis testing.
Section~\ref{section: summary} concludes the paper, followed by a
discussion on the MDP that is relegated to an appendix.

\section{Concentration and Its Relation to the Moderate Deviations Principle}
\label{section: Concentration Inequalities via Martingales}

We present here some essential material that is related to the
martingale approach used in this paper for the moderate-deviations
analysis of binary hypothesis testing. A background on martingales
is provided in, e.g., \cite{Williams} where we only rely here on
basic knowledge on martingales.

%\subsection{Doob's Martingales} \label{subsection: Martingales}
%This sub-section provides a short background on martingales to set
%definitions and notation. For a more thorough study of
%martingales, the reader it referred to, e.g., \cite{Billingsley}.
%
%\begin{definition}{\bf[Doob's Martingale]} Let $(\Omega, \mathcal{F},
%\pr)$ be a probability space. A Doob's martingale sequence is a
%sequence $X_0, X_1, \ldots$ of random variables (RVs) and
%corresponding sub $\sigma$-algebras $\mathcal{F}_0, \mathcal{F}_1,
%\ldots$ (also denoted by $\{X_i, \mathcal{F}_i\}$) that satisfy
%the following conditions:
%\begin{enumerate}
%\item $X_i \in \LL^1(\Omega, \mathcal{F}_i, \pr)$ for every $i$, i.e., each $X_i$
%is defined on the same sample space $\Omega$, it is measurable
%with respect to the corresponding $\sigma$-algebra $\mathcal{F}_i$
%(i.e., $X_i$ is $\mathcal{F}_i$-measurable) and $\expectation
%[|X_i|] = \int_{\Omega} |X_i(\omega)| d\pr(\omega) < \infty.$
%\item $\mathcal{F}_0 \subseteq \mathcal{F}_1 \subseteq \ldots $ (where this sequence
%of $\sigma$-algebras is called a filtration).
%\item $X_i = \expectation[ X_{i+1} | \mathcal{F}_i]$
%holds almost surely (a.s.) for every $i$.
%\end{enumerate}
%\label{definition: Doob's martingales}
%\end{definition}

\subsection{Azuma's Inequality} \label{subsection:
Azuma's inequality} Azuma's inequality\footnote{Azuma's inequality
is also known as the Azuma-Hoeffding inequality. It will be named
from this point as Azuma's inequality for the sake of brevity.}
forms a useful concentration inequality for bounded-difference
martingales \cite{Azuma}. In the following, this inequality is
introduced. The reader is referred to, e.g., \cite{survey2006}
and \cite{McDiarmid_tutorial} for
surveys on concentration inequalities for martingales (including a
proof of this inequality).

\begin{theorem}{\bf[Azuma's inequality]}
Let $\{X_k, \mathcal{F}_k\}_{k=0}^{\infty}$ be a
discrete-parameter real-valued martingale sequence (where
$\mathcal{F}_0 \subseteq \mathcal{F}_1 \subseteq \ldots$ is called
a filtration). Assume that for every $k \in \naturals$, the
condition $ |X_k - X_{k-1}| \leq d_k$ holds a.s. for some
non-negative constants $\{d_k\}_{k=1}^{\infty}$. Then
\begin{equation}
\pr( | X_n - X_0 | \geq r) \leq 2 \exp\left(-\frac{r^2}{2
\sum_{k=1}^n d_k^2}\right) \, \quad \forall \, r \geq 0.
\label{eq: Azuma's concentration inequality - general case}
\end{equation}
\label{theorem: Azuma's concentration inequality}
\end{theorem}

The concentration inequality stated in Theorem~\ref{theorem:
Azuma's concentration inequality} was proved in \cite{Hoeffding}
for independent bounded random variables, and it was later derived
in \cite{Azuma} for bounded-difference martingales. %This
%inequality has been extensively used so far in
%information-theoretic aspects of modern coding theory (see
%\cite{RiU_book} and references therein).

\subsection{A Refined Version of Azuma's Inequality}

\begin{theorem}
Let $\{X_k, \mathcal{F}_k\}_{k=0}^{\infty}$ be a
discrete-parameter real-valued martingale. Assume that, for some
constants $d, \sigma > 0$, the following two requirements are
satisfied a.s.
\begin{eqnarray*}
&& | X_k - X_{k-1} | \leq d, \\
&& \text{Var} (X_k | \mathcal{F}_{k-1}) = \expectation \bigl[(X_k
- X_{k-1})^2 \, | \, \mathcal{F}_{k-1} \bigr] \leq \sigma^2
\end{eqnarray*}
for every $k \in \{1, \ldots, n\}$. Then, for every $\alpha \geq
0$,
\begin{equation}
\hspace*{-0.2cm} \pr(|X_n-X_0| \geq \alpha n) \leq 2 \exp\left(-n
\, D\biggl(\frac{\delta+\gamma}{1+\gamma} \Big|\Big|
\frac{\gamma}{1+\gamma}\biggr) \right) \label{eq: first refined
concentration inequality}
\end{equation}
where
\begin{equation}
\gamma \triangleq \frac{\sigma^2}{d^2}, \quad \delta \triangleq
\frac{\alpha}{d}  \label{eq: notation}
\end{equation}
and $D(p || q) \triangleq p \ln\Bigl(\frac{p}{q}\Bigr) + (1-p)
\ln\Bigl(\frac{1-p}{1-q}\Bigr)$ for $p, q \in [0,1]$ is the
divergence (a.k.a. relative entropy or Kullback-Leibler distance)
between the two probability distributions $(p,1-p)$ and $(q,1-q)$.
If $\delta>1$, then the probability on the left-hand side of
\eqref{eq: first refined concentration inequality} is equal to
zero. \label{theorem: first refined concentration inequality}
\end{theorem}
\begin{proof}
See \cite{McDiarmid_bounded_differences_Martingales_1989},
\cite[Corollary~2.4.7]{Dembo_Zeitouni} or \cite[Section~III]{Sason_submitted_paper}.
\end{proof}

\subsection{Relation of Theorem~\ref{theorem: first refined concentration inequality}
with the Moderate Deviations Principle for i.i.d. RVs}
\label{subsection: MDP for real-valued i.i.d. RVs}

According to the moderate deviations theorem (see, e.g.,
\cite[Theorem~3.7.1]{Dembo_Zeitouni}) in $\reals$, let
$\{X_i\}_{i=1}^n$ be a sequence of i.i.d. real-valued RVs such
that $\Lambda_X(\lambda) = \expectation[e^{\lambda X_i}] < \infty$
in some neighborhood of zero, and also assume that
$\expectation[X_i] = 0$ and $\sigma^2 = \text{Var}(X_i) > 0$. Let
$\{a_n\}_{n=1}^{\infty}$ be a non-negative sequence such that $a_n
\rightarrow 0$ and $n a_n \rightarrow \infty$ as $n \rightarrow
\infty$, and let
\begin{equation}
Z_n \triangleq \sqrt{\frac{a_n}{n}} \sum_{i=1}^n X_i, \quad
\forall \, n \in \naturals. \label{eq: Z sequence}
\end{equation}
Then, for every measurable set $\Gamma \subseteq \reals$,
\begin{eqnarray}
&& -\frac{1}{2 \sigma^2} \inf_{x \in \Gamma^0} x^2 \nonumber\\[0.1cm]
&& \leq \liminf_{n \rightarrow \infty} a_n \ln \pr(Z_n \in \Gamma) \nonumber \\
&& \leq \limsup_{n \rightarrow \infty}
a_n \ln \pr(Z_n \in \Gamma)  \nonumber \\
&& \leq -\frac{1}{2 \sigma^2} \inf_{x \in \overline{\Gamma}} x^2
\end{eqnarray}
where $\Gamma^0$ and $\overline{\Gamma}$ designate, respectively,
the interior and closure sets of $\Gamma$.

Let $\eta \in (\frac{1}{2}, 1)$ be an arbitrary fixed number, and
let $\{a_n\}_{n=1}^{\infty}$ be the non-negative sequence
$$a_n = n^{1-2\eta}, \quad \forall \, n \in \naturals$$ so that
$a_n \rightarrow 0$ and $n a_n \rightarrow \infty$ as $n
\rightarrow \infty$. Let $\alpha \in \reals^+$, and $\Gamma
\triangleq (-\infty, -\alpha] \cup [\alpha, \infty)$. Note that,
from \eqref{eq: Z sequence},
$$ \pr\left( \Big|\sum_{i=1}^n X_i \Big| \geq \alpha n^{\eta} \right)
= \pr(Z_n \in \Gamma)$$ so from the moderate deviations principle
(MDP)
\begin{equation}
\hspace*{-0.4cm} \lim_{n \rightarrow \infty} n^{1-2\eta} \; \ln
\pr\left( \Big|\sum_{i=1}^n X_i \Big| \geq \alpha n^{\eta} \right)
= -\frac{\alpha^2}{2 \sigma^2}, \; \; \forall \, \alpha \geq 0.
\label{eq: MDP for i.i.d. real-valued RVs}
\end{equation}
It is demonstrated in Appendix~\ref{appendix: MDP} that, in
contrast to Azuma's inequality, Theorem~\ref{theorem: first
refined concentration inequality} gives an upper bound on the
probability $\pr\left( \Big|\sum_{i=1}^n X_i \Big| \geq \alpha
n^{\eta} \right)$ (where $n \in \naturals$ and $\alpha \geq 0$)
which coincides with the exact asymptotic limit in \eqref{eq: MDP
for i.i.d. real-valued RVs}. The analysis in
Appendix~\ref{appendix: MDP} provides another interesting link
between Theorem~\ref{theorem: first refined concentration
inequality} and a classical result in probability theory, which
also emphasizes the significance of the refinements of Azuma's
inequality.

\section{Moderate Deviations Analysis for Binary Hypothesis Testing}
\label{section: binary hypothesis testing}

Binary hypothesis testing for finite alphabet models was analyzed
via the method of types, e.g., in \cite[Chapter~11]{Cover and
Thomas} and \cite{Csiszar_Shields_FnT}. It is assumed that the
data sequence is of a fixed length $(n)$, and one wishes to make
the optimal decision based on the received sequence and the
Neyman-Pearson ratio test.

Let the RVs $X_1, X_2 ....$ be i.i.d. $\sim Q$, and consider two
hypotheses:
\begin{itemize}
\item $H_1:  Q = P_1$.
\item $H_2:  Q = P_2$.
\end{itemize}
For the simplicity of the analysis, let us assume that the RVs are
discrete, and take their values on a finite alphabet $\mathcal{X}$
where $P_1(x), P_2(x) > 0$ for every $x \in \mathcal{X}$.

In the following, let
\begin{equation*}
L(X_1, \ldots, X_n) \triangleq \ln \frac{P_1^n(X_1, \ldots,
X_n)}{P_2^n(X_1, \ldots, X_n)} = \sum_{i=1}^n \ln
\frac{P_1(X_i)}{P_2(X_i)}
\end{equation*}
designate the log-likelihood ratio. By the strong law of large
numbers (SLLN), if hypothesis $H_1$ is true, then a.s.
\begin{equation}
\lim_{n \rightarrow \infty} \frac{L(X_1, \ldots, X_n)}{n} = D(P_1
|| P_2) \label{eq: a.s. limit of the normalized LLR under
hypothesis H1}
\end{equation}
and otherwise, if hypothesis $H_2$ is true, then a.s.
\begin{equation}
\lim_{n \rightarrow \infty} \frac{L(X_1, \ldots, X_n)}{n} = -D(P_2
|| P_1) \label{eq: a.s. limit of the normalized LLR under
hypothesis H2}
\end{equation}
where the above assumptions on the probability mass functions
$P_1$ and $P_2$ imply that the relative entropies, $D(P_1 || P_2)$
and $D(P_2 || P_1)$, are both finite. Consider the case where for
some fixed constants $\overline{\lambda}, \underline{\lambda} \in
\reals$ that satisfy $$-D(P_2||P_1) < \underline{\lambda} \leq
\overline{\lambda} < D(P_1||P_2)$$ one decides on hypothesis $H_1$
if $ L(X_1, \ldots, X_n) > n \overline{\lambda} $, and on
hypothesis $H_2$ if $ L(X_1, \ldots, X_n) < n
\underline{\lambda}.$ Note that if $\overline{\lambda} =
\underline{\lambda} \triangleq \lambda$ then a decision on the two
hypotheses is based on comparing the normalized log-likelihood
ratio (w.r.t. $n$) to a single threshold $(\lambda)$, and deciding
on hypothesis $H_1$ or $H_2$ if this normalized log-likelihood
ratio is, respectively, above or below $\lambda$. If
$\underline{\lambda} < \overline{\lambda}$ then one decides on
$H_1$ or $H_2$ if the normalized log-likelihood ratio is,
respectively, above the upper threshold $\overline{\lambda}$ or
below the lower threshold $\underline{\lambda}$. Otherwise, if the
normalized log-likelihood ratio is between the upper and lower
thresholds, then an erasure is declared and no decision is taken
in this case.

Let
\begin{eqnarray}
&& \alpha_n^{(1)} \triangleq P_1^n \Bigl( L(X_1, \ldots, X_n) \leq
n \overline{\lambda} \Bigr)
\label{eq: error and erasure event under hypothesis H1} \\
&& \alpha_n^{(2)} \triangleq P_1^n \Bigl( L(X_1, \ldots, X_n) \leq
n \underline{\lambda} \Bigr) \label{eq: error event under
hypothesis H1}
\end{eqnarray}
and
\begin{eqnarray}
&& \beta_n^{(1)}  \triangleq P_2^n \Bigl( L(X_1, \ldots, X_n) \geq
n \underline{\lambda} \Bigr)
\label{eq: error and erasure event under hypothesis H2} \\
&& \beta_n^{(2)}  \triangleq P_2^n \Bigl( L(X_1, \ldots, X_n) \geq
n \overline{\lambda} \Bigr) \label{eq: error event under
hypothesis H2}
\end{eqnarray}
then $\alpha_n^{(1)}$ and $\beta_n^{(1)}$ are the probabilities of
either making an error or declaring an erasure under,
respectively, hypotheses $H_1$ and $H_2$; similarly
$\alpha_n^{(2)}$ and $\beta_n^{(2)}$ are the probabilities of
making an error under hypotheses $H_1$ and $H_2$, respectively.

Let $\pi_1, \pi_2 \in (0,1)$ denote the a-priori probabilities of
the hypotheses $H_1$ and $H_2$, respectively, so
\begin{equation}
P_{\text{e}, n}^{(1)} = \pi_1 \alpha_n^{(1)} + \pi_2 \beta_n^{(1)}
\label{eq: overall probability of a mixed error and erasure event}
\end{equation}
is the probability of having either an error or an erasure, and
\begin{equation}
P_{\text{e}, n}^{(2)} = \pi_1 \alpha_n^{(2)} + \pi_2 \beta_n^{(2)}
\label{eq: overall error probability}
\end{equation}
is the probability of error.

Based on the asymptotic results in \eqref{eq: a.s. limit of the
normalized LLR under hypothesis H1} and \eqref{eq: a.s. limit of
the normalized LLR under hypothesis H2}, which hold a.s. under
hypotheses $H_1$ and $H_2$ respectively, the large deviations
analysis refers to upper and lower thresholds $\overline{\lambda}$
and $\underline{\lambda}$ which are {\em kept fixed} (i.e., these
thresholds do not depend on the block length $n$ of the data
sequence) where
$$ -D(P_2 || P_1) < \underline{\lambda} \leq  \overline{\lambda} <
D(P_1 || P_2).$$ Suppose that instead of having some fixed upper
and lower thresholds, one is interested to set these thresholds
such that as the block length $n$ tends to infinity, they tend
simultaneously to their asymptotic limits in \eqref{eq: a.s. limit
of the normalized LLR under hypothesis H1} and \eqref{eq: a.s.
limit of the normalized LLR under hypothesis H2}, i.e.,
$$ \lim_{n \rightarrow \infty} \overline{\lambda}^{(n)} = D(P_1 ||
P_2), \quad \lim_{n \rightarrow \infty} \underline{\lambda}^{(n)}
= -D(P_2 || P_1).$$ Specifically, let $\eta \in (\frac{1}{2}, 1)$,
and $\varepsilon_1, \varepsilon_2 > 0$ be arbitrary fixed numbers,
and consider the case where one decides on hypothesis $H_1$ if
$L(X_1, \ldots, X_n) > n \overline{\lambda}^{(n)}$, and on
hypothesis $H_2$ if $L(X_1, \ldots, X_n) < n
\underline{\lambda}^{(n)}$ where these upper and lower thresholds
are set to
\begin{eqnarray*}
&& \overline{\lambda}^{(n)} = D(P_1 || P_2) - \varepsilon_1
n^{-(1-\eta)} \\
&& \underline{\lambda}^{(n)} = -D(P_2 || P_1) + \varepsilon_2
n^{-(1-\eta)}
\end{eqnarray*}
so that they approach, respectively, the relative entropies $D(P_1
|| P_2)$ and $-D(P_2 || P_1)$ in the asymptotic case where the
block length $n$ of the data sequence tends to infinity.
Accordingly, the conditional probabilities in \eqref{eq: error and
erasure event under hypothesis H1}--\eqref{eq: error event under
hypothesis H2} are modified so that the fixed thresholds
$\overline{\lambda}$ and $\underline{\lambda}$ are replaced with
the above block-length dependent thresholds
$\overline{\lambda}^{(n)}$ and $\underline{\lambda}^{(n)}$,
respectively. The moderate deviations analysis for binary
hypothesis testing studies the probability of an error event and
the probability of a joint error and erasure event under the two
hypotheses, and it studies the interplay between each of these
probabilities, the block length $n$, and the related thresholds
that tend asymptotically to the limits in \eqref{eq: a.s. limit of
the normalized LLR under hypothesis H1} and \eqref{eq: a.s. limit
of the normalized LLR under hypothesis H2} when the block length
tends to infinity.

In light of the discussion in Section~\ref{subsection: MDP for
real-valued i.i.d. RVs} on the MDP for i.i.d. RVs and the
discussion of its relation to Theorem~\ref{theorem: first refined
concentration inequality} (see Appendix~\ref{appendix: MDP}), and
also motivated by the three recent works in \cite[Section~4.3]{Abbe_thesis},
\cite{AltugW_ISIT2010}
and \cite{He_IT09}, we proceed to consider in the following
moderate deviations analysis for binary hypothesis testing. Our
approach for this kind of analysis is different, and it relies on
concentration inequalities for martingales.

In the following, we analyze the probability of a joint error and
erasure event under hypothesis $H_1$, i.e., derive an upper bound
on $\alpha_n^{(1)}$ in \eqref{eq: error and erasure event under
hypothesis H1}. The same kind of analysis can be adapted easily
for the other probabilities in \eqref{eq: error event under
hypothesis H1}--\eqref{eq: error event under hypothesis H2}.

Under hypothesis $H_1$, let us construct the martingale sequence
$\{U_k, \mathcal{F}_k\}_{k=0}^n$ where $\mathcal{F}_0 \subseteq
\mathcal{F}_1 \subseteq \ldots \mathcal{F}_n$ is the filtration
$$ \mathcal{F}_0 = \{\emptyset, \Omega\}, \quad \mathcal{F}_k =
\sigma(X_1, \ldots, X_k), \; \; \forall \, k \in \{1, \ldots,
n\}$$ and
\begin{equation}
U_k = \expectation_{P_1^n} \bigl[ L(X_1, \ldots, X_n) \; | \;
\mathcal{F}_k \bigr]. \label{eq: martingale sequence U under
hypothesis H1}
\end{equation}
For every $k \in \{0, \ldots, n\}$
\begin{eqnarray*}
&& U_k = \expectation_{P_1^n} \Biggl[  \sum_{i=1}^n
\ln \frac{P_1(X_i)}{P_2(X_i)} \; \Big| \; \mathcal{F}_k  \Biggr] \\
&& \hspace*{0.5cm} =  \sum_{i=1}^k \ln \frac{P_1(X_i)}{P_2(X_i)} +
\sum_{i=k+1}^n \expectation_{P_1^n} \Biggl[
\ln \frac{P_1(X_i)}{P_2(X_i)} \Biggr] \\
&& \hspace*{0.5cm} =  \sum_{i=1}^k \ln \frac{P_1(X_i)}{P_2(X_i)} +
(n-k) D(P_1 || P_2).
\end{eqnarray*}
In particular
\begin{eqnarray}
&& U_0 = n D(P_1 || P_2), \label{eq: initial value of the
martingale U that is related to the binary hypothesis testing}   \\
&& U_n = \sum_{i=1}^n \ln \frac{P_1(X_i)}{P_2(X_i)} = L(X_1,
\ldots, X_n) \label{eq: final value of the martingale U that is
related to the binary hypothesis testing}
\end{eqnarray}
and, for every $k \in \{1, \ldots, n\}$,
\begin{equation}
U_k - U_{k-1} = \ln \frac{P_1(X_k)}{P_2(X_k)} - D(P_1 || P_2).
\label{eq: jumps of the martingale U that is related to the binary
hypothesis testing}
\end{equation}
Let
\begin{equation}
d_1 \triangleq \max_{x \in \mathcal{X}} \left| \ln
\frac{P_1(x)}{P_2(x)} - D(P_1 || P_2) \right| \label{eq: d1}
\end{equation}
so $d_1 < \infty$ since by assumption the alphabet set
$\mathcal{X}$ is finite, and $P_1(x), P_2(x) > 0$ for every $x \in
\mathcal{X}$. From \eqref{eq: jumps of the martingale U that is
related to the binary hypothesis testing} and \eqref{eq: d1},
$|U_k - U_{k-1}| \leq d_1$ a.s. for every $k \in \{1, \ldots,
n\}$, and due to the statistical independence of $\{X_i\}$
\begin{eqnarray}
&& \hspace*{-1cm} \expectation_{P_1^n} \bigl[ (U_k - U_{k-1})^2
\, | \, \mathcal{F}_{k-1} \bigr] \nonumber \\
&& \hspace*{-1cm} = \sum_{x \in \mathcal{X}} \left\{ P_1(x) \left(
\ln \frac{P_1(x)}{P_2(x)} - D(P_1 || P_2) \right)^2 \right\}
\triangleq \sigma_1^2. \label{eq: sigma1 squared for the jumps of
the martingale U}
\end{eqnarray}

Let $\varepsilon_1 > 0$ and $\eta \in (\frac{1}{2}, 1)$ be
two arbitrarily fixed numbers. Then, under hypothesis~$H_1$, it
follows from Theorem~\ref{theorem: first refined concentration
inequality} and the above construction of a martingale that
\begin{eqnarray}
&& P_1^n\bigl( L(X_1, \ldots, X_n) \leq n
\overline{\lambda}^{(n)})
\nonumber \\
&& = P_1^n\bigl( U_n - U_0 \leq -\varepsilon_1 n^\eta \bigr) \nonumber \\
&& \leq \exp \left( -n D\biggl(\frac{\delta_1^{(\eta, n)} +
\gamma_1}{1+\gamma_1} \, \big|\big| \,
\frac{\gamma_1}{1+\gamma_1}\biggr) \right) \label{eq:  1st
inequality for the moderate deviations analysis of binary
hypothesis testing}
\end{eqnarray}
where
\begin{equation}
\delta_1^{(\eta, n)} \triangleq \frac{\varepsilon_1
n^{-(1-\eta)}}{d_1}, \quad \gamma_1 \triangleq
\frac{\sigma_1^2}{d_1^2} \label{delta1 and gamma1 for the moderate
deviations analysis of binary hypothesis testing}
\end{equation}
with $d_1$ and $\sigma_1^2$ from \eqref{eq: d1} and \eqref{eq:
sigma1 squared for the jumps of the martingale U}.

In the following, we will make use of the following lemma:
\begin{lemma}
\begin{equation}
(1+u) \ln(1+u) \geq \left\{
\begin{array}{ll}
u + \frac{u^2}{2}, \quad & u \in [-1, 0] \\[0.2cm]
u+\frac{u^2}{2}-\frac{u^3}{6}, \quad & u \geq 0
\end{array}
\right. \label{eq: inequality for lower bounding the divergence}
\end{equation}
where at $u=-1$, the left-hand side is defined to be zero (it is
the limit of this function when $u \rightarrow -1$ from above).
\label{lemma: inequality for lower bounding the divergence}
\end{lemma}
\begin{proof}
The proof follows by elementary calculus.
\end{proof}

From \eqref{delta1 and gamma1 for the moderate deviations analysis
of binary hypothesis testing} and the inequality in
Lemma~\ref{lemma: inequality for lower bounding the divergence},
it follows that
\begin{eqnarray*}
&& D\biggl(\frac{\delta_1^{(\eta, n)} + \gamma_1}{1+\gamma_1} \,
\big|\big| \, \frac{\gamma_1}{1+\gamma_1}\biggr) \\
%&& = \frac{\gamma_1}{1+\gamma_1} \left[ \Bigl(1 +
%\frac{\delta_1^{(\eta, n)}}{\gamma_1} \Bigr) \ln \Bigl(1 +
%\frac{\delta_1^{(\eta, n)}}{\gamma_1}
%\Bigr) \right. \\
%&& \hspace*{1.5cm} \left. + \frac{\bigl(1-\delta_1^{(\eta,
%n)}\bigr) \ln\bigl(1-\delta_1^{(\eta, n)}\bigr)}{\gamma_1} \right] \\
&& \geq \frac{\gamma_1}{1+\gamma_1} \left[
\biggl(\frac{\delta_1^{(\eta, n)}}{\gamma_1} +
\frac{\bigl(\delta_1^{(\eta, n)}\bigr)^2}{2 \gamma_1^2} -
\frac{\bigl(\delta_1^{(\eta, n)}\bigr)^3}{6 \gamma_1^3}
\biggr) \right. \\
&& \hspace*{1.5cm} \left. + \frac{1}{\gamma_1}
\biggl(-\delta_1^{(\eta, n)} + \frac{(\delta_1^{(\eta, n)})^2}{2}
\biggr) \right] \\
&& = \frac{\bigl(\delta_1^{(\eta, n)}\bigr)^2}{2 \gamma_1} -
\frac{\bigl(\delta_1^{(\eta, n)}\bigr)^3}{6 \gamma_1^2
(1+\gamma_1)} \\[0.1cm]
%&& = \frac{\varepsilon_1^2 \, n^{-2(1-\eta)}}{2 \gamma_1 d_1^2}
%\left( 1 - \frac{\varepsilon_1}{3 d_1 \gamma_1 (1+\gamma_1)}
%\, \frac{1}{n^{1-\eta}} \right) \\[0.1cm]
&& = \frac{\varepsilon_1^2 \, n^{-2(1-\eta)}}{2 \sigma_1^2} \left(
1 - \frac{\varepsilon_1 d_1}{3 \sigma_1^2 (1+\gamma_1)} \,
\frac{1}{n^{1-\eta}} \right)
\end{eqnarray*}
provided that $\delta_1^{(\eta, n)} < 1$ (which holds for $n \geq
n_0$ for some $n_0 \triangleq n_0(\eta, \varepsilon_1, d_1) \in
\naturals$ that is determined from \eqref{delta1 and gamma1 for
the moderate deviations analysis of binary hypothesis testing}).
By substituting this lower bound on the divergence into \eqref{eq:
1st inequality for the moderate deviations analysis of binary
hypothesis testing}, it follows that
\begin{eqnarray}
&& \hspace*{-1.7cm} \alpha_n^{(1)} = P_1^n\bigl( L(X_1, \ldots,
X_n) \leq n D(P_1 || P_2) - \varepsilon_1 n^\eta \bigr) \nonumber \\
&& \hspace*{-1.0cm} \leq \exp \left(-\frac{\varepsilon_1^2 \,
n^{2\eta-1}}{2 \sigma_1^2} \left( 1 - \frac{\varepsilon_1 d_1}{3
\sigma_1^2 (1+\gamma_1)} \, \frac{1}{n^{1-\eta}} \right) \right).
\label{eq: 2nd inequality for the moderate deviations analysis of
binary hypothesis testing}
\end{eqnarray}
Consequently, in the limit where $n$ tends to infinity,
\begin{equation}
\lim_{n \rightarrow \infty} n^{1-2\eta} \ln \, \alpha_n^{(1)} \leq
-\frac{\varepsilon_1^2}{2 \sigma_1^2} \label{eq: 3rd inequality
for the moderate deviations analysis of binary hypothesis testing}
\end{equation}
with $\sigma_1^2$ in \eqref{eq: sigma1 squared for the jumps of
the martingale U}. From the analysis in Section~\ref{subsection:
MDP for real-valued i.i.d. RVs} and Appendix~\ref{appendix: MDP},
it follows that the inequality for the asymptotic limit in
\eqref{eq: 3rd inequality for the moderate deviations analysis of
binary hypothesis testing} holds in fact with equality. To verify
this, consider the real-valued sequence of i.i.d. RVs
$$ Y_i \triangleq \ln \left( \frac{P_1(X_i)}{P_2(X_i)} \right) - D(P_1 || P_2),
\quad i=1, \dots, n $$ that, under hypothesis $H_1$, have zero
mean and variance $\sigma_1^2$. Since, by assumption, the sequence
$\{X_i\}_{i=1}^n$ are i.i.d., then
\begin{equation}
L(X_1, \ldots, X_n) - n D(P_1 || P_2) = \sum_{i=1}^n Y_i,
\label{eq: equality related to the LLR and the sequence Y}
\end{equation}
and it follows from the one-sided version of the MDP in \eqref{eq:
MDP for i.i.d. real-valued RVs} that indeed \eqref{eq: 3rd
inequality for the moderate deviations analysis of binary
hypothesis testing} holds with equality. Moreover,
Theorem~\ref{theorem: first refined concentration inequality}
provides, via the inequality in \eqref{eq: 2nd inequality for the
moderate deviations analysis of binary hypothesis testing}, a
finite-length result that enhances the asymptotic result for $n
\rightarrow \infty$.

In the considered setting of moderate deviations analysis for
binary hypothesis testing, the upper bound on the probability
$\alpha_n^{(1)}$ in \eqref{eq: 2nd inequality for the moderate
deviations analysis of binary hypothesis testing}, which refers to
the probability of either making an error or declaring an erasure
(i.e., making no decision) under the hypothesis $H_1$, decays to
zero sub-exponentially with the length $n$ of the sequence. As
mentioned above, based on the analysis in Section~\ref{subsection:
MDP for real-valued i.i.d. RVs} and Appendix~\ref{appendix: MDP},
the asymptotic upper bound in \eqref{eq: 3rd inequality for the
moderate deviations analysis of binary hypothesis testing} is
tight. A completely similar moderate-deviations analysis can be
also performed under the hypothesis $H_2$. Hence, a
sub-exponential scaling of the probability $\beta_n^{(1)}$ in
\eqref{eq: error and erasure event under hypothesis H2} of either
making an error or declaring an erasure (where the lower threshold
$\underline{\lambda}$ is replaced with
$\underline{\lambda}^{(n)}$) also holds under the hypothesis
$H_2$. These two sub-exponential decays to zero for the
probabilities $\alpha_n^{(1)}$ and $\beta_n^{(1)}$, under
hypothesis $H_1$ or $H_2$ respectively, improve as the value of
$\eta \in (\frac{1}{2}, 1)$ is increased. On the other hand, the
two {\em exponential decays} to zero of the probabilities of error
(i.e., $\alpha_n^{(2)}$ and $\beta_n^{(2)}$ under hypothesis $H_1$
or $H_2$, respectively) improve as the value of $\eta \in
(\frac{1}{2}, 1)$ is decreased; this is due to the fact that, for
a fixed value of $n$, the margin which serves to protect us from
making an error (either under hypothesis $H_1$ or $H_2$) is
increased by decreasing the value of $\eta$ as above (note that by
reducing the value of $\eta$ for a fixed $n$, the upper and lower
thresholds $\overline{\lambda}^{(n)}$ and
$\underline{\lambda}^{(n)}$ are made closer to $D(P_1||P_2)$ from
below and to $-D(P_2||P_1)$ from above, respectively, which
therefore increases the margin that is used for protecting one
from making an erroneous decision). This shows the existence of a
tradeoff, in the choice of the parameter $\eta \in (\frac{1}{2},
1)$, between the probability of error and the joint probability of
error and erasure under either hypothesis $H_1$ or $H_2$ (where
this tradeoff exists symmetrically for each of the two
hypotheses).

In \cite{AltugW_ISIT2010} and
\cite{Polyanskiy_Verdu_Allerton2010}, the authors consider
moderate deviations analysis for channel coding over memoryless
channels. In particular, \cite[Theorem~2.2]{AltugW_ISIT2010} and
\cite[Theorem~6]{Polyanskiy_Verdu_Allerton2010} indicate on a
tight lower bound (i.e., a converse) to the asymptotic result in
\eqref{eq: 3rd inequality for the moderate deviations analysis of
binary hypothesis testing} for binary hypothesis testing. This
tight converse is indeed consistent with the asymptotic result of
the MDP in \eqref{eq: MDP for i.i.d. real-valued RVs} for
real-valued i.i.d. random variables, which implies that the
asymptotic upper bound in \eqref{eq: 3rd inequality for the
moderate deviations analysis of binary hypothesis testing},
obtained via the martingale approach with the refined version of
Azuma's inequality in Theorem~\ref{theorem: first refined
concentration inequality}, holds indeed with equality. Note that
this equality does not follow from Azuma's inequality, so its
refinement was essential for obtaining this equality. The reason
is that, due to Appendix~\ref{appendix: MDP}, the upper bound in
\eqref{eq: 3rd inequality for the moderate deviations analysis of
binary hypothesis testing} that is equal to
$-\frac{\varepsilon_1^2}{2 \sigma_1^2}$ is replaced via Azuma's
inequality by the looser bound $-\frac{\varepsilon_1^2}{2 d_1^2}$
(note that, from \eqref{eq: d1} and \eqref{eq: sigma1 squared for
the jumps of the martingale U}, $\sigma_1 \leq d_1$ where in
general $\sigma_1$ may be significantly smaller than $d_1$).

\section{Summary}
\label{section: summary} This paper is focused on the moderate
deviations analysis of binary hypothesis testing. The analysis is based on a concentration
inequality for discrete-parameter martingales with bounded jumps,
which forms a refined version of Azuma's inequality (see
\cite[Corollary~2.4.7]{Dembo_Zeitouni}). The relation of this
concentration inequality to the moderate deviations principle for
i.i.d. random variables is considered. This paper presents in
part the work in \cite{Sason_submitted_paper}, and it exemplifies
the use of a refinement of Azuma's inequality in an
information-theoretic aspect. Further information-theoretic
applications are considered in, e.g., \cite{Sason_Aachen}
and \cite{kostis_ISIT12}. The slides are available in \cite{slides}.

\vspace*{0.2cm}
{\em{Acknowledgment}}: One of the reviewers pointed
out that the moderate deviations analysis in this work
can be done alternatively by relying on results, e.g.,
from \cite{Arkhangel'skii} or \cite{Rozovsky}.
We thank the reviewer for this note, and we currently
study this line of work.

\appendices

\section{Analysis Related to the Moderate Deviations Principle For
i.i.d. RVs (See Section~\ref{subsection: MDP for real-valued
i.i.d. RVs})} \label{appendix: MDP}

It is demonstrated in the following that, in contrast to Azuma's
inequality, Theorem~\ref{theorem: first refined concentration
inequality} provides an upper bound on $\pr\left(
\Big|\sum_{i=1}^n X_i \Big| \geq \alpha n^{\eta} \right)$ for
$\alpha \geq 0$, which coincides with the correct asymptotic
result in \eqref{eq: MDP for i.i.d. real-valued RVs}. It is proved
under the further assumption that there exists some constant $d >
0$ such that $|X_k| \leq d$ a.s. for every $k \in \naturals$
(since the RVs $\{X_k\}$ are assumed to be i.i.d., it is
sufficient to require it for $k=1$). Let us define the martingale
sequence $\{S_k, \mathcal{F}_k\}_{k=0}^n$ where $S_k \triangleq
\sum_{i=1}^k X_i$ and $\mathcal{F}_k \triangleq \sigma(X_1,
\ldots, X_k)$ for every $k \in \{1, \ldots, n\}$ with $S_0 = 0$
and $\mathcal{F}_0 = \{\emptyset, \mathcal{F}\}$.

\subsubsection{Analysis related to Azuma's inequality}
The martingale sequence $\{S_k, \mathcal{F}_k\}_{k=0}^n$ has
uniformly bounded jumps, where $|S_k - S_{k-1}| = |X_k| \leq d$
a.s. for every $k \in \{1, \ldots, n\}$. Hence it follows from
Azuma's inequality that, for every $\alpha \geq 0$,
\begin{equation*}
\pr\left( |S_n| \geq \alpha n^{\eta} \right) \leq 2
\exp\left(-\frac{\alpha^2 n^{2\eta-1}}{2d^2}\right)
\end{equation*}
and therefore
\begin{equation}
\lim_{n \rightarrow \infty} n^{1-2 \eta} \; \ln \pr\bigl( |S_n|
\geq \alpha n^{\eta} \bigr) \leq -\frac{\alpha^2}{2d^2}.
\label{eq: MDP scaling from Azuma's inequality for the sum of
i.i.d. real-valued RVs}
\end{equation}
This differs from the limit in \eqref{eq: MDP for i.i.d.
real-valued RVs} where $\sigma^2$ is replaced by $d^2$, so Azuma's
inequality does not provide the correct asymptotic result in
\eqref{eq: MDP for i.i.d. real-valued RVs} (unless $\sigma^2 =
d^2$, i.e., $|X_k|=d$ a.s. for every $k$).

\subsubsection{Analysis related to Theorem~\ref{theorem: first refined concentration inequality}}
From Theorem~\ref{theorem: first refined concentration
inequality}, it follows that for every $\alpha \geq 0$,
\begin{equation*}
\pr(|S_n| \geq \alpha n^{\eta}) \leq 2 \exp\left(-n \,
D\biggl(\frac{\delta'+\gamma}{1+\gamma} \Big|\Big|
\frac{\gamma}{1+\gamma}\biggr) \right)
\end{equation*}
where $\gamma$ is introduced in \eqref{eq: notation}, and
$\delta'$ is given by
\begin{equation}
\delta' \triangleq \frac{\frac{\alpha}{n^{1-\eta}}}{d} = \delta
n^{-(1-\eta)} \label{eq: new delta'}
\end{equation}
due to the definition of $\delta$ in \eqref{eq: notation}. Hence,
it follows that
\begin{eqnarray*}
&& \pr(|S_n| \geq \alpha n^{\eta}) \\
&& \leq 2 \exp\left( -\frac{\delta^2 n^{2\eta-1}}{2\gamma} \left[
1 + \frac{\alpha (1-\gamma)}{3 \gamma d} \cdot n^{-(1-\eta)} +
\ldots \right] \right)
\end{eqnarray*}
for every $n \in \naturals$, and therefore (since, from \eqref{eq:
notation}, $\frac{\delta^2}{\gamma} = \frac{\alpha^2}{\sigma^2})$
\begin{equation}
\lim_{n \rightarrow \infty} n^{1-2 \eta} \; \ln \pr\bigl( |S_n|
\geq \alpha n^{\eta} \bigr) \leq -\frac{\alpha^2}{2 \sigma^2}.
\label{eq: MDP scaling from Theorem 2 for the sum of i.i.d.
real-valued RVs}
\end{equation}
Hence, this bound coincides with the exact limit
in \eqref{eq: MDP for i.i.d. real-valued RVs}.


\begin{thebibliography}{99}
\bibitem{Abbe_thesis}
E. A. Abbe, {\em Local to Global Geometric Methods in Information
Theory}, Ph.D. dissertation, MIT, Boston, MA, USA, June 2008.
\bibitem{AlonS_tpm3}
N. Alon and J. H. Spencer, {\em The Probabilistic Method}, Wiley
Series in Discrete Mathematics and Optimization, Third Edition,
2008.
\bibitem{Arkhangel'skii}
A. N. Arkhangel'skii, ``Lower bounds for probabilities of large
deviations for sums of independent random variables,'' {\em
Theory of Probability and Applications}, vol.~34, no.~4, pp.~565-575,
1989.
\bibitem{AltugW_ISIT2010}
Y. Altu\v{g} and A. B. Wagner, ``Moderate deviations analysis of
channel coding: discrete memoryless case,'' {\em Proceedings 2010
IEEE International Symposium on Information Theory (ISIT~2010)},
pp.~265--269, Austin, Texas, USA, June~2010.
\bibitem{Azuma}
K. Azuma, ``Weighted sums of certain dependent random variables,''
{\em Tohoku Mathematical Journal}, vol.~19, pp.~357--367, 1967.
%\bibitem{Billingsley}
%P. Billingsley, {\em Probability and Measure}, Wiley Series in
%Probability and Mathematical Statistics, Third Edition, 1995.
%\bibitem{Blahut_IT74}
%R. E. Blahut, ``Hypothesis testing and information theory,'' {\em
%IEEE Trans. on Information Theory}, vol.~20, no.~4, pp.~405--417,
%July 1974.
%\bibitem{Chung_LU2006}
%F. Chung and L. Lu, {\em Complex Graphs and Networks}, {\em
%Regional Conference Series in Mathematics}, vol.~107, 2006.
%Chapter 2 entitled ``Old and new concentration inequalities.''
%[Online]. Available:
%\url{http://www.math.ucsd.edu/~fan/complex/ch2.pdf}.
\bibitem{survey2006}
F. Chung and L. Lu, ``Concentration inequalities and martingale
inequalities: a survey,'' {\em Internet Mathematics}, vol.~3,
no.~1, pp.~79--127, March 2006.% [Online]. Available:
%\url{http://www.ucsd.edu/~fan/wp/concen.pdf}.
\bibitem{Cover and Thomas}
T. M. Cover and J. A. Thomas, {\em Elements of Information
Theory}, John Wiley and Sons, second edition, 2006.
\bibitem{Csiszar_Shields_FnT}
I. Csisz\'{a}r and P. C. Shields, {\em Information Theory and
Statistics: A Tutorial}, Foundations and Trends in Communications
and Information Theory, vol.~1, no.~4, pp.~417--528, 2004.
\bibitem{Dembo_paper96}
A. Dembo, ``Moderate deviations for martingales with bounded
jumps,'' {\em Electronic Communications in Probability}, vol.~1,
no.~3, pp.~11--17, March 1996.
\bibitem{Dembo_Zeitouni}
A. Dembo and O. Zeitouni, {\em Large Deviations Techniques and
Applications}, Springer, second edition, 1997.
%\bibitem{DubashiP09_book}
%D. P. Dubashi and A. Panconesi, {\em Concentration of Measure for
%the Analysis of Randomized Algorithms}, Cambridge University
%Press, 2009.
\bibitem{He_IT09}
D. He, L. A. Lastras-Monta\~{n}o, E. Yang, A. Jagmohan and J.
Chen, ``On the redundancy of Slepian-Wolf coding,'' {\em IEEE
Trans. on Information Theory}, vol.~55, no.~12, pp.~5607--5627,
December 2009.
\bibitem{Hoeffding}
W. Hoeffding, ``Probability inequalities for sums of bounded
random variables,'' {\em Journal of the American Statistical
Association}, vol.~58, no.~301, pp.~13--30, March 1963.
\bibitem{Hollander_book_2000}
F. den Hollander, {\em Large Deviations}, Fields Institute
Monographs, American Mathematical Society, 2000.
\bibitem{McDiarmid_bounded_differences_Martingales_1989}
C. McDiarmid, ``On the method of bounded differences,''
{\em Surveys in Combinatorics}, vol.~141, pp.~148--188, Cambridge
University Press, Cambridge, 1989.
\bibitem{McDiarmid_tutorial}
C. McDiarmid, ``Concentration,'' {\em Probabilistic Methods for
Algorithmic Discrete Mathematics}, pp.~195--248, Springer, 1998.
\bibitem{Polyanskiy_Poor_Verdu_IT2010}
Y. Polyanskiy, H. V. Poor, and S. Verd\'{u}, ``Channel coding rate
in finite blocklength regime,'' {\em IEEE Trans. on Information
Theory}, vol.~56, no.~5, pp.~2307--2359, May~2010.
\bibitem{Polyanskiy_Verdu_Allerton2010}
Y. Polyanskiy and S. Verd\'{u}, ``Channel dispersion and moderate
deviations limits of memoryless channels,'' Proceedings
Forty-Eighth Annual Allerton Conference,  pp.~1334--1339, UIUC,
Illinois, USA, October 2010.
\bibitem{Rozovsky}
L. V. Rozovsky, ``Estimate from below for large-deviation probabilities
of a sum of independent random variables with finite variances,'' {\em Journal
of Mathematical Sciences}, vol.~109, no.~6, May 2002.
%\bibitem{Rosenthal}
%J. S. Rosenthal, {\em A First Look at Rigorous Probability
%Theory}, World Scientific Publishes, second edition, 2006.
%\bibitem{RiU_book}
%T. J. Richardson and R. Urbanke, {\em Modern Coding Theory},
%Cambridge University Press, 2008.
\bibitem{Sason_submitted_paper}
I. Sason, ``On refined versions of the Azuma-Hoeffding inequality
with applications in information theory,'' last updated in July 2012. [Online].
Available: \url{http://arxiv.org/pdf/1111.1977v5.pdf}.
\bibitem{Sason_Aachen}
I. Sason, ``On the concentration of the crest factor for OFDM
signals,'' {\em Proceedings of the 2011 8th International
Symposium on Wireless Communication Systems (ISWCS '11)},
pp.~784--788, Aachen, Germany, November 2011. [Online]. Available:
\url{http://arxiv.org/abs/1111.1982}.
\bibitem{slides}
I. Sason, "On Concentration and moderate deviations analysis of binary hypothesis testing,''
presentation is online available at \url{http://webee.technion.ac.il/people/sason/ISIT2012a_presentation.pdf}.
\bibitem{Tan_arxiv11}
V. Y. F. Tan, ``Moderate-deviations of lossy source coding for
discrete and Gaussian sources,''
\url{http://arxiv.org/abs/1111.2217}, November~2011.
\bibitem{Williams}
D. Williams, {\em Probability with Martingales}, Cambridge University
Press, 1991.
\bibitem{kostis_ISIT12}
K. Xenoulis, N. Kalouptsidis and I. Sason, ``New achievable rates for nonlinear Volterra
channels via martingale inequalities,'' {\em Proceedings of the 2012 IEEE International
Symposium of Information Theory}, pp.~1430--1434, MIT, Boston, USA, July 2012.
\end{thebibliography}
\end{document}